\documentclass[aps,%
 reprint,
 amsmath,amssymb,
 aps,
]{revtex4-2}
\usepackage{amsthm}
\usepackage{thmtools}
\usepackage{graphicx}
\usepackage{dcolumn}
\usepackage{bm}

\usepackage{xcolor} 
\usepackage[
  colorlinks=true,
  linkcolor=blue,   
  citecolor=blue,   
  urlcolor=blue     
]{hyperref}

\usepackage{diagbox}
\usepackage{etoolbox}
\usepackage{dsfont}
\usepackage{cleveref}
\usepackage[normalem]{ulem}

\AtBeginDocument{%
}

\usepackage{array,makecell,graphicx,ragged2e} 
 
\newcommand{\mcell}[2][1.7em]{\makebox[#1][c]{$#2$}} 
\newcommand{\catprod}[2]{\left(\mcell[1.6em]{#1}\!\otimes\!\mcell[1.6em]{#2}\right)}
\newcommand{\entry}[3]{\mcell[1.6em]{#1}\!\otimes\!\catprod{#2}{#3}}

\newcommand{\ket}[1]{\left| #1 \right\rangle}
\newcommand{\bra}[1]{\left\langle #1 \right|}

\newcommand{\vect}[1]{\mathbf{#1}}

\newtheorem{definition}{Definition}
\newtheorem{result}{Result}

\newtheorem{corollary}{Corollary}

\begin{document}


\title{Catalytic Activation of Genuine Multipartite Entanglement and Nonlocality}

\author{Eliot Donnadieu}
\author{Pavel Sekatski}
\author{Nicolas Brunner}
\author{Victor Barizien}
\affiliation{
 Department of Applied Physics, University of Geneva, Geneva, Switzerland}
 \affiliation{
 Geneva Quantum Centre, University of Geneva, Geneva, Switzerland}




\date{\today}
\raggedbottom
\begin{abstract}
We demonstrate the possibility to activate genuine multipartite entanglement (GME), the strongest form of entanglement for multipartite states, within the framework of quantum catalysis. Specifically, we show that any biseparable state (i.e.~not GME) that is not partition separable can be deterministically transformed into a GME state via the help of a catalyst and local operations, without any classical communication. In turn, we construct a catalytic protocol tailored to the multipartite case. The protocol is termed ``sum-to-product'', as it transforms a mixture of states into their tensor product in a heralded manner. We apply this protocol to random network entangled states, which are biseparable by construction, and demonstrate catalytic activation of both GME and genuine multipartite Bell nonlocality. 
\end{abstract}

\maketitle

\section{introduction}
Multipartite entanglement is a key concept in quantum physics, both from a conceptual and applied perspective \cite{Horodecki_2009,Friis2018}. Its characterization is therefore an important task, which is made challenging by the existence of many different forms of multipartite entanglement \cite{Guhne2009,horodecki24}. Readily adapting the concept from the bipartite case, an $N$-partite quantum state is said to be entangled if it is not fully separable, i.e.~it cannot be written as a convex mixture of product states over $N$ parties. Remarkably, multipartite systems can exhibit a much stronger form of entanglement: a state is called genuinely multipartite entangled (GME)~\cite{Seevinck2001} if it is not biseparable, meaning that it cannot be written as a convex mixture of states that are separable with respect to some bipartition of the $N$ parties.
GME has emerged as an important resource for multipartite quantum processing tasks, including quantum computation~\cite{Bruss11}, quantum metrology~\cite{Toth12}, and quantum key distribution in networks~\cite{Epping17}. 

From a resource-theoretic perspective, the notion of GME is defined with respect to the paradigm of local operation and classical communication, similarly to the bipartite case. Specifically, no stochastic local operation and classical communication (SLOCC) protocols can create GME starting from biseparable states, see e.g.~\cite{Ma25,horodecki24}. Surprisingly, however, it has been shown that some biseparable states can become GME when combining several copies~\cite{Yamasaki_2022}. In fact, this is possible for all states that are not partition-separable~\cite{Palazuelos22}, i.e.~not separable with respect to any fixed bipartition. This highlights a fundamental difference between GME and bipartite (or standard multipartite) entanglement, since separable states remain separable for any number of copies. This effect of superactivation of GME has attracted growing interest recently \cite{Weinbrenner2025,Baksova2025,Starek2026,baksová2026,zhang26}.

Here, we uncover another form a activation of GME via the process of quantum catalysis \cite{catalysis,Datta2023}. We show that a single-copy of a biseparable state, combined with an auxiliary system (the catalyst), can be transformed into a GME state while the catalyst is returned unchanged. Importantly, this can be done deterministically with only local operations, and no  classical communication. First, we show that this effect is generic: GME can be activated for any bisparable that is not partition separable via catalytic local operations (CLO) \cite{clod}. Moreover, the activation of GME is even possible using a catalyst that is itself biseparable. Then, we develop an efficient catalytic protocol (within CLO) tailored to the multipartite case. We apply this protocol to a class of mixed states that we term ``random network states'', where pure entanglement is shared stochastically between subsets of parties. Beyond the activation of GME, we show that this protocol leads to activation of an even stronger form of quantum correlations, namely genuine multipartite Bell nonlocality (GMNL)~\cite{Svetlichny87,Bancal13}, a resource for multipartite device-independent protocols \cite{Ribeiro2018,GMNL2}. To this end, we establish a technical result, which might be of independent interest, showing that GMNL distributions remain GMNL under mixing with distributions producing complementary (i.e.~orthogonal) outputs.

\section{catalytic activation of GME}\label{section2}

Here, we work within the framework of correlated catalysis under local operations~\cite{clod}. In this context, $N$ parties share a state $\rho_S$, and in addition they have access to a $N$-partite catalyst in the state $\omega_C$. By means of local operations, i.e.~a product channel $\Lambda = \bigotimes_{i=1}^{N}\Lambda_i$, they transform their systems to the global state $ \tau_{S'C}:=\Lambda(\rho_S\otimes\omega_C),$
with
\begin{equation}\label{cdt}
    \text{Tr}_{S'}\left(\tau_{S'C}\right) = \omega_C \quad \text{and} \quad \text{Tr}_C\left(\tau_{S'C}\right) = \tau_{S'}.
\end{equation}
The left equality in Eq.~\eqref{cdt} expresses the fact that the catalyst is returned unchanged, nevertheless potentially correlated with the output system $S'$, i.e.~$\tau_{S'C} \neq \tau_{S'} \otimes \omega_C$. Importantly, the final system state $\tau_{S'}$ can contain a resource that $\rho_S$ did not have initially. Here, we will show that the final state $\tau_{S'}$ can be GME starting from an initial biseparable state $\rho_S$.

\begin{definition}
    An $N$-partite quantum state $\rho_{A_1\cdots A_N}$, shared between parties $A_1\cdots A_N=:S$, is GME~\cite{Seevinck2001} if it is not biseparable, i.e.~if there exists no decomposition of the form
\begin{equation}
    \rho_S
    =
    \sum_{\mathcal G|\bar{\mathcal G}}
    \sum_i
    p_{i,\mathcal G}\,
    \rho_{\mathcal G}^{(i)}
    \otimes
    \rho_{\bar{\mathcal G}}^{(i)} ,
\end{equation}
where the first sum runs over all bipartitions $\mathcal G|\bar{\mathcal G}$ of the parties, $p_{i,\mathcal G}\geq 0$, and $\sum_{\mathcal G|\bar{\mathcal G}}\sum_i p_{i,\mathcal G}=1$.
\end{definition}
Remarkably, it was shown recently that any biseparable state $\rho_S$ which is not partition-separable becomes GME in the multiple-copy scenario~\cite{Yamasaki_2022,Palazuelos22}. More precisely, if there is no bi-partition $\mathcal{G}|\bar{\mathcal{G}}$ for which $\rho_S$ can be written as
\begin{equation}\label{eq: part-sep}
\rho_S
=
\sum_i p_i\,
\rho_{\mathcal{G}}^{(i)}
\otimes
\rho_{\bar{\mathcal{G}}}^{(i)},
\end{equation}
then there exists a finite number of copies \(k\) such that $\rho^{\otimes k}_{S}$
is GME. This multi-copy property has a direct application within the catalytic framework. Specifically, the CLO protocol recently presented in Ref.~\cite{bavaresco}, which we refer to as “multi-copy catalysis”, achieves the probabilistic (but heralded) transformation of a state into multiple copies of itself. This protocol, described below, is an adaptation of the catalytic protocol discussed in~Refs.~\cite{Duan2005,Lipka2021,Wilming2021,Ganardi2024}.

More precisely, for any biseparable state $\rho_S$, there exists a set of local operations $\Lambda$ and an $N$-partite catalyst state $\omega_C$ such that the catalytic conditions given in Eq.~\eqref{cdt} are satisfied, and
\begin{equation}\label{eq: taus}
    \tau_{S'}
    =
    \frac{1}{k} \rho^{\otimes k}\otimes[0]
    +
    \frac{k-1}{k} \sigma^{\otimes k}\otimes [1],
\end{equation}
where $\sigma_{S}:=\sigma_{A_1}\otimes \cdots \otimes \sigma_{A_N}$ is an arbitrary product state (hence locally preparable by the parties), and 
\begin{equation}
 [i]_S=\ket{i}\bra{i}_{A_1}\otimes\cdots\otimes \ket{i}\bra{i}_{A_N}
\end{equation}
denotes a classically correlated ``flag'' state. The final state $\tau_{S'}$ is thus a ``flagged mixture'' of the multi-copy state $\rho_{S}^{\otimes k}$, and a fully separable product state $\sigma^{\otimes k}$. When $\rho_{S}$ is not partition separable, there exists a finite $k$ such that $\rho_{S}^{\otimes k}$ is GME. It follows that $\tau_{S'}$ is GME since the reading of the local flags is an SLOCC operation, which cannot create GME. Therefore, any biseparable $\rho_{S}$ that is not partition separable can be catalytically transformed into a GME state.  Notice that this activation can be done with a catalyst that is itself in a biseparable state. For this, consider the minimal $k$ such that $\rho^{\otimes k}$ is GME. In this case the catalyst leading to Eq.~\eqref{eq: taus} contains terms with at most $k-1$ copies of $\rho$ and is biseparable.


Conversely, note that CLO cannot activate GME from partition-separable states. Indeed, the squashed entanglement~\cite{christandl02squashed} is a faithful bipartite entanglement measure~\cite{Brandao2011} (i.e.~it is zero iff the state is separable), which cannot be increased across the bipartition given in Eq.~\eqref{eq: part-sep} under Catalytic Local Operations and Classical Communication (CLOCC)~\cite{clod}. Thus partition-separable states remain partition separable under CLO. The above discussion is summarized in the following result.

\begin{result}\label{result1}
    A biseparable state can be transformed into a GME state by CLO if and only if it is not partition-separable. Moreover, this can be achieved with a biseparable catalyst.
\end{result}

This result is surprisingly strong: indeed, while bipartite entanglement cannot be activated catalytically, CLO is sufficient to activate GME as long as the initial state is not partition-separable.
In particular, note that no classical communication is required.

While it is always possible to activate GME using multi-copy catalysis, it is natural to ask wether there exists other CLO protocols that are more efficient in terms of production rate and resources requirements (size and complexity of the catalyst). Another direction is to consider the activation of stronger forms of multipartite resources than GME, such as GMNL. In the following we address both questions, starting with the former.

\section{sum-to-product catalysis\label{newcatalysis}}

In this section, we present a catalytic protocol that transforms any flagged mixture of quantum states
\begin{equation}\label{mixturegeneral}
     \rho_S=\sum_i \alpha_i \rho_i\otimes [i]
 \end{equation}
with weights $\alpha_i\geq0$, such that $\sum_i \alpha_i =1$, into a flagged mixture involving a branch $\bigotimes_i\rho_i$ with all states present simultaneously. 

\begin{restatable}{theorem}{THMOne}\label{lem1}
    Any $N$-partite state of the form of Eq.~\eqref{mixturegeneral} can be transformed by CLO into
\begin{equation}
\tau_{S'}=\frac{1}{C_\alpha}\bigotimes_{k=1}^n \rho_k\otimes [1]+\sum_{i=2}^{n^2}p_i\gamma_i\otimes[i],
\label{target}
\end{equation}
where $p_i \geq 0$, $\sum_{i=2}^{n^2}p_i=\left(1-\frac{1}{C_\alpha}\right)$, $\gamma_i$ are quantum states on $S'$,  and $C_\alpha := \sum_{i=1}^n 1/\alpha_i$. The state of the catalyst is given by
\begin{equation}
    \omega_C = \frac{1}{C_\alpha} \sum_{j=1}^{n}  \frac{1}{\alpha_j}
\sigma^{\otimes (j-1)}\otimes \left(
\bigotimes_{k=j+1}^n \rho_k \right) \otimes [j],
\label{catalyst}
\end{equation}
with any product state $\sigma=\sigma_{C_1}\otimes\cdots\otimes\sigma_{C_N}$.
\end{restatable}

\begin{proof}
We present here the case $n=3$ with uniform weights $\alpha_i = 1/3$, while the general case follows the same idea and can be found in Appendix~\ref{appendix1}. The initial system state is
\begin{equation}
    \rho_S=\frac{1}{3}(\rho_1\otimes[1]+\rho_2\otimes[2]+\rho_3\otimes[3]),
\end{equation} 
while the catalyst is
\begin{equation}
    \omega_{C}=\frac{1}{3}\left(\rho_2\otimes \rho_3\otimes[1]+\sigma \otimes \rho_3\otimes[2]+\sigma \otimes \sigma \otimes[3]\right).\!
    \label{cat3}
\end{equation}
The state $\rho_{S}\otimes \omega_C$ is therefore a mixture of 9 elements, as represented in Table~\ref{cat3step1}. 
We now proceed to construct the set of local operations $\Lambda$ leading to the desired transformation. The first step is to enlarge the system by locally preparing two product states $\sigma\otimes\sigma$. The remaining operations are only made of local swaps and flags relabeling, as described below in steps (i)-(iv).
 \begin{table}[t]
\setlength{\tabcolsep}{6pt}
\renewcommand{\arraystretch}{1.7} 
\small 
\begin{ruledtabular}
\begin{tabular*}{\columnwidth}{@{\extracolsep{\fill}}ccc}
\multicolumn{1}{c}{\shortstack{Syst$\otimes$(Cat)}} &
\multicolumn{1}{c}{\shortstack{Syst$\otimes$(Cat)}} &
\multicolumn{1}{c}{\shortstack{Syst$\otimes$(Cat)}}\\
\hline
$\entry{\rho_1}{\rho_2}{\rho_3}$ &
$\entry{\rho_2}{\rho_2}{\rho_3}$ &
$\entry{\rho_3}{\rho_2}{\rho_3}$\\[4pt]
$\entry{\rho_1}{\sigma}{\rho_3}$ &
$\entry{\rho_2}{\sigma}{\rho_3}$ &
$\entry{\rho_3}{\sigma}{\rho_3}$\\[4pt]
$\entry{\rho_1}{\sigma}{\sigma}$ &
$\entry{\rho_2}{\sigma}{\sigma}$ &
$\entry{\rho_3}{\sigma}{\sigma}$
\end{tabular*}
\end{ruledtabular}
\caption{Table representation of the state $\rho_S\otimes \omega_C$.
Each entry gives the system-catalyst state corresponding to the flag state $[i]_S\otimes[j]_C$. The target system state $\rho_1 \otimes \rho_2 \otimes \rho_3$ is obtained at entry $[1]_S\otimes[1]_C$ by using $\rho_2 \otimes \rho_3$ from the catalyst. For some of the other entries, the state of the system is send to the catalyst to recover its state on average.}
\label{cat3step1}
\end{table}

\noindent \textbf{(i)} When parties read flag $[1]_S\otimes [1]_C$, they build the tensor product state $\rho_1\otimes \rho_2\otimes \rho_3$ by swapping $\sigma\otimes\sigma$ with the state $\rho_2\otimes \rho_3$ of the catalyst, 
 \begin{equation}
     \rho_1\otimes\sigma\otimes \sigma\otimes(\rho_2\otimes \rho_3)\rightarrow \rho_1\otimes\rho_2\otimes \rho_3\otimes(\sigma\otimes \sigma).
 \end{equation}
After that step, the tensor product branch of the final system state is reached. Additional operations are necessary to recover the catalyst in the second marginal. To do so, one needs to recover 3 times each elements 
of the catalyst. Note that the 3 states $\sigma \otimes \sigma$ are already recovered in entries $[1]\otimes[1]$ and $[x] \otimes [3]$ for $x\neq 3$.

\noindent \textbf{(ii)}~Entries $[1]_S\otimes [x]_C$ for $x\neq1$ already give 2 states $\rho_2\otimes\rho_3$ in the catalyst subspace. Therefore no actions are required when parties measure these flags. 
The third state $\rho_2\otimes\rho_3$ is recovered in the diagonal entry $[2]_S\otimes[2]_C$ by swapping $\rho_2$ system's state and $\sigma$ of the catalyst.

\noindent \textbf{(iii)}~Entries $[2]_S\otimes [x]_C$ for $x\neq2$ again give 2 states $\rho_3\otimes\sigma$ for the catalyst. The third state $\rho_3\otimes\sigma$ is recovered in the diagonal entry by swapping $\rho_3$ system's state and $\sigma$ of the catalyst.

\noindent \textbf{(iv)} At each step, parties add a copy of the catalyst flag in the system $[i]_S\otimes[j]_C\rightarrow [i]_{S'}\otimes[j]_{S'}\otimes [j]_C$ and perform the cyclic relabel of the diagonal entries $[i]_{S'}\otimes [i]_{S'}\otimes[i]_C\rightarrow [i-1]_{S'}\otimes [i-1]_{S'}\otimes[i-1]_C$ with $[0]\equiv[3]$.

After these operations, the final state $\tau_{S'C}$ is such that the partial trace $\tau_{S'} :=\text{Tr}_C(\tau_{S'C})$ is (up to a cyclic flag relabeling) of the form given in Eq.~\eqref{target}, whereas the marginal $\text{Tr}_{S'}(\tau_{S'C})$ is exactly the catalyst given in Eq.~\eqref{cat3}. The catalytic marginal conditions in Eq.~\eqref{cdt} are thus satisfied.
\end{proof}

This protocol can be viewed as a generalization of multi-copy catalysis. Indeed, if all $\rho_i$ in Eq.~\eqref{mixtinit} are the same, our catalyst becomes precisely the one introduced in Theorem~1 of Ref.~\cite{bavaresco}. More generally, it is interesting to compare the two protocols in terms of resources. Indeed, the tensor product state~$\bigotimes_i \rho_i$ can also be obtained probabilistically by taking $n$ copies of the state in Eq.~\eqref{mixtinit}, and thus also by using the multi-copy catalysis. 
However, this is much more costly in terms of the size of the catalyst, as it requires $\mathcal{O}(n^{n-1})$ states, instead of $\mathcal{O}\left(\frac{n^2}{2}\right)$ states for our catalyst in Eq.~\eqref{catalyst}. Moreover, the multi-copy catalyst produces the target tensor product state with probability $\frac{n!}{n^{n+1}}$, while our protocol  allows for a production rate scaling as $\frac{1}{n^2}$, a significant improvement. Our catalytic protocol is thus more efficient for both figures of merit. Finally, our protocol can work with a partition-separable catalyst, in contrast to the multi-copy catalysis, see Appendix~\ref{exampleapendix}.

A surprising property of our protocol is that the catalyst in Eq.~\eqref{catalyst} is not unique. As shown in Appendix~\ref{appendix1}, there exists a continuous set of valid catalysts, generated by the convex mixture of $n!$ canonical catalysts, constructed from permutations of the $n$ states $\rho_i$'s.

Another property of our protocol is that it can be concatenated, i.e.~two such protocols can be applied sequentially to define another CLO protocol. This property is not a priori guaranteed, as such concatenation may in general correlate the two catalysts.

\begin{figure*}
    \centering
    \includegraphics[trim=5.5cm 7.5cm 6.5cm 6cm, clip=true, width=.9\linewidth]{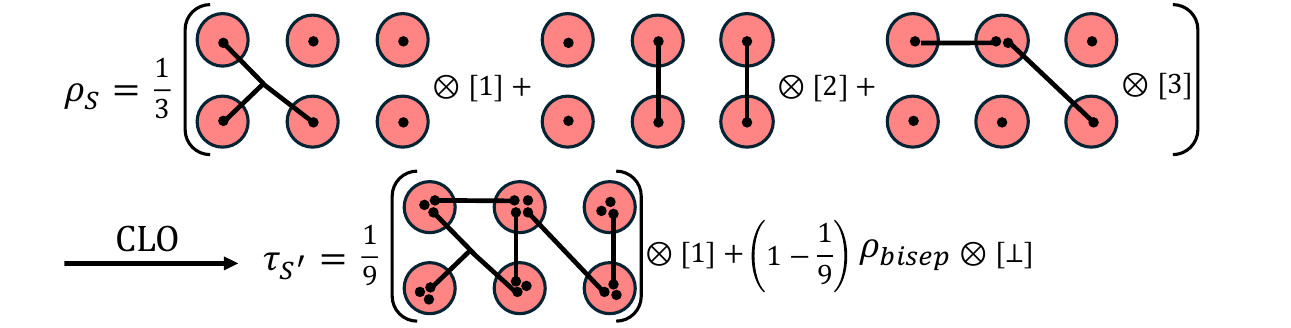}
    \caption{Graphical representation of an initial $6$-partite random network state $\rho_S$ with equal weights, where solid lines represent pure entanglement, dark dots are quantum systems, and red circles are parties. This initial state is a flagged mixture of: (i) a GHZ-type state~\cite{Greenberger90} shared between three parties
    ; (ii) two pure bipartite entangled states shared between two groups of two parties; (iii) two pure bipartite entangled states shared between three parties. This state is transformed by CLO into a final state $\tau_{S'}$ that is a flagged mixture of an entangled connected network state and a biseparable state. Accordingly to Result~\ref{result1}, the state $\tau_{S'}$ is GME. For completeness, the catalyst state used for this transformation is given in Appendix~\ref{app:catalystfig1}.}
    \label{image}
\end{figure*}

\begin{restatable}{theorem}{THMtwo}
\label{lem2}
    Consider a CLO $(\omega_{01},\Lambda_{01})$ that enables the transformation
    \begin{equation}
        \rho_0\rightarrow p_1\rho_1\otimes [1]+\sum_{i>1}p_i\sigma_0^{(i)}\otimes[i]
    \end{equation}
     and a second one $(\omega_{12},\Lambda_{12})$ that allows to transform
     \begin{equation}
         \rho_1\rightarrow q_1\rho_2\otimes [1]+\sum_{i>1}q_i\sigma_1^{(i)}\otimes[i]
     \end{equation}
     where the catalysts $\omega_{01}, \omega_{12}$ are of the form $\sum_i\eta_i\omega_i\otimes[i]$. 
    
    Then there exists a CLO $(\omega_{01}\otimes\omega_{12},\Lambda_{02})$ that transforms $\rho_0$ into $p_1q_1\rho_2\otimes[1]+\sum_{i>1}\xi_i\sigma_2^{(i)}\otimes[i]$.
\end{restatable}

The proof of this theorem is given in Appendix~\ref{appendix2}. The theorem will be used later for the activation of GMNL.

\section{Catalytic activation of random network states}

We now apply the catalytic protocol to demonstrate activation of GME and GMNL. To this end, we focus on the following class of states.

\begin{definition}
    A multipartite random network states is given by
    \begin{equation}\label{mixtinit}
    \rho_{S}=\sum_{i=1}^n \alpha_i \,\Psi_i\otimes [i],
\end{equation}
where each $\Psi_i$ is pure and GME on one or several subsets of parties; see Fig.~\ref{image} for an example.
\end{definition} 

Such states occur in multipartite networks in which different subsets of parties are randomly connected through shared sources of entanglement. 
In particular, each subset of parties connected by some $\Psi_i$ can be represented by a graph where edges are drawn between all parties appearing in the subset. If the union of these graphs (edges) is not fully connected, the state $\rho_S$ is partition separable. Indeed, the parties can then be partitioned in two groups $\mathcal G|\bar{\mathcal G}$ with no edges between $\mathcal G$ and $\bar{\mathcal G}$, i.e.~no entanglement distributed across the bipartition. 
Conversely, we prove in the following that if the union of these graphs is fully connected, $\rho_S$ becomes GME and even GMNL (see~\Cref{def: GMNL}) via the CLO protocols introduced in Section~\ref{newcatalysis}.

\begin{result}
    \label{theorem1}
    Any random network state (Eq.~\ref{mixtinit}) that is not partition separable exhibits CLO activation of GME via the sum-to-product protocol.
\end{result}

Indeed, applying Theorem~\ref{lem1} to any flagged mixture as in Eq.~\eqref{mixtinit} results in
\begin{equation}
\tau_{S'}=\frac{1}{n^2}\Psi_{\text{ECN}}\otimes[1]+\left(1-\frac{1}{n^2}\right)\gamma \otimes [\perp],
\end{equation}
where $\gamma$ is an $N$-partite state, one has performed the flag relabelling $[1]_{S'}\rightarrow [1]_{S'}$ and $[i]\rightarrow[\perp]$ $\forall i\neq 0$, with $[\perp]$ an orthogonal flag to $[1]$, and $\Psi_{\text{ECN}}=\bigotimes_i \Psi_i$ is a pure entangled connected network (ECN) state, which are GME~\cite{GMNL2}. Hence, $\tau_{S'}$ is also GME. For an illustrative example, see Fig.~\ref{image}.

To go further, we show that our new protocols can also activate the stronger resource of GMNL. Indeed, GME is necessary for GMNL, but not sufficient~\cite{Augusiak15,Bowles16genuinely}.

\begin{definition} \label{def: GMNL}
    A state is said to be GMNL if it can be locally measured to generate a probability distribution $p(\mathbf{a}|\mathbf{x})=\text{Tr}(M_{A_1}^{a_1|x_1}\otimes\cdots\otimes M_{A_N}^{a_N|x_N}\rho_{A_1\cdots A_N})$ which is not decomposable as a convex mixture of local distributions with respect to bi-partitions, i.e.~that $p(\mathbf{a}|\mathbf{x})$ cannot be written as
\begin{equation}\label{GMNLdistrib}
p(\mathbf{a}|\mathbf{x})=\sum_{\mathcal{G}}\eta_{\mathcal{G}}\sum_\lambda\mu_\mathcal{G}(\lambda)p_{\mathcal{G}}(\mathbf{a}_{\mathcal{G}}|\mathbf{x}_{\mathcal{G}},\lambda)p_{\bar{\mathcal{G}}}(\mathbf{a}_{\bar{\mathcal{G}}}|\mathbf{x}_{\bar{\mathcal{G}}},\lambda),
\end{equation}
where $\textbf{a}=(a_1,\cdots,a_N)\quad \text{and}\quad \textbf{x}=(x_1,\cdots,x_N)$, $\sum_{\mathcal{G}, \lambda}\mu_{\mathcal{G}}(\lambda)=1$, $\sum_{\mathcal{G}}\eta_{\mathcal{G}}=1$, and $\mu_{\mathcal{G}}(\lambda),\eta_\mathcal{G}\geq 0$ for all $\mathcal{G}, \lambda$. Here, we use the definition of Ref.~\cite{Bancal13}, where all $p_{\mathcal{G}}$ are assumed to be no-signaling.
\end{definition} 
Now, we use the concatenation of catalytic protocols  for the activation of GMNL. 

\begin{result}
    \label{theorem2}
    Any random network state (Eq.~\ref{mixtinit}) that is not partition separable exhibits CLO activation of GMNL.
\end{result}

To prove this statement, we concatenate the previous CLO protocol with multi-copy catalysis using Theorem~\ref{lem2} to obtain the final state
\begin{equation}
    \tilde \tau_{S'}=\frac{1}{kn^2}\Psi_\text{GMNL}\otimes[1]+\left(1-\frac{1}{kn^2}\right)\tilde{\gamma}\otimes[\perp],
\end{equation}
where $\tilde{\gamma}$ is a quantum state, and $\Psi_\text{GMNL}=\Psi_{\text{ECN}}^{\otimes k}$ is GMNL for $k$ large enough. Indeed, any $N$-partite pure ECN state exhibits multi-copy GMNL activation~\cite{GMNL2}, i.e.~there exists a finite $k\leq N-1$ such that $\Psi_{\text{ECN}}^{\otimes k}$ is GMNL. It remains to be shown that $\tilde \tau_{S'}$ is itself GMNL. Since $\Psi_\text{GMNL}$ violates a GMNL inequality, there exists a set of measurements with outcomes $\{1,\cdots,m\}$ that leads to a GMNL distribution $p_\text{GMNL}(\vect{a},\vect{x})$. To prove GMNL of $\tilde \tau_{S'}$, we build a new GMNL distribution in the following way: if the flag $[1]_{S'}$ is measured, parties perform the measurements that leads to $p_\text{GMNL}$; otherwise, all parties always output an orthogonal inconclusive outcome $\varnothing$. This enlarges the set of possible outcomes to $\{1,\cdots,m\}\cup \{\varnothing\}$. Therefore, the full probability distribution is a mixture $p(\vect{a}|\vect{x})=\frac{1}{kn^2}\,p_\text{GMNL}(\mathbf{a}|\mathbf{x})+(1-\frac{1}{kn^2})\delta_\varnothing(\mathbf{a}|\mathbf{x})$ for some distribution $\delta_{\varnothing}$ having only $\varnothing$ as an output. Crucially, the following theorem ensures that this new distribution is still GMNL, a result that might be of independent interest.

\begin{restatable}{theorem}{THMThree}\label{thm2}
    Any multipartite distribution with local outcomes in $M \cup \bar{M}$, where $M$ and $\bar{M}$ are two orthogonal set of outcomes, of the form
    \begin{equation}\label{taufinal}
        p(\vect{a}|\vect{x})=\kappa\,p_\text{GMNL}(\vect{a}\in \bm M|\mathbf{x})+(1-\kappa)p_\perp(\mathbf{a}\in\bar{\bm M}|\mathbf{x}), 
    \end{equation}
    with $0<\kappa\leq 1$ is GMNL.
\end{restatable}

The proof, detailed in Appendix~\ref{lemtechniqueAp} proceeds by contradiction. Supposing that $p$ is not GMNL and thus admits a decomposition in the form of Eq.~\eqref{GMNLdistrib}, we construct such a decomposition for $p_{\text{GMNL}}$. This construction is possible because the distributions $p_{\text{GMNL}}$ and $p_{\perp}$ have locally orthogonal outputs for all parties. This allows us to show, using the no-signaling assumption, that all product distributions $p_{\mathcal{G}}(\mathbf{a}_{\mathcal{G}}|\mathbf{x}_{\mathcal{G}},\lambda)p_{\bar{\mathcal{G}}}(\mathbf{a}_{\bar{\mathcal{G}}}|\mathbf{x}_{\bar{\mathcal{G}}},\lambda)$ in the decomposition of $p$ as in Eq.~\eqref{GMNLdistrib} must have all their outputs in either $M$ or $\bar {M}$ for each $\lambda$. 

Strikingly, an appropriate choice of catalyst for the procedure can lead to activation of GMNL with partition-separable catalysts, as illustrated in Appendix~\ref{exampleapendix}.

\section{Conclusion and outlook}

In this work, we demonstrate that GME can be activated via catalytic local operations. The process is deterministic and relies solely on local operations, without involving any form of communication. In particular, we show that this effect is generic: all states are GME under CLO if and only if they are not partition-separable. In turn, we introduce a ``sum-to-product'' catalytic protocol, transforming a flagged mixture of quantum states into their tensor product in a heralded manner. We apply this protocol to random network states, demonstrating CLO activation of GME and GMNL. 

This activation can be particularly strong. Consider a random network states involving only bipartite entanglement. Such state is not only biseparable, but is in fact $N-1$ separable (hence minimal in terms of entanglement depth). Yet, it can be catalytically transformed into a GME and GMNL state using a catalyst that is partition separable. 

More generally, our catalytic protocol could find other applications. For example, consider the effect of CHSH activation~\cite{Navascues11}: two states $\rho_1$ and $\rho_2$, both not violating CHSH, can be combined into $\rho_1 \otimes \rho_2$, which violates CHSH. Starting from the flagged mixture of $\rho_1$ and $\rho_2$ (which clearly does not violate CHSH), one will get CHSH violation catalytically. 

Our work complements a research line investigating the power of catalytic operations for boosting entanglement and other forms of quantum correlations such as Bell nonlocality. It is well known that catalysis can enable state transformations (notably via LOCC) which would otherwise be impossible \cite{jonathan99}, as strikingly illustrated by the phenomenon of embezzlement \cite{vanDamHayden2003}, or increasing the performance in protocols such as teleportation \cite{Lipka2021}. A common feature of these protocols is that they transform an initial resource state into a stronger one, for example turning a weakly entangled state into a strongly entangled one. In contrast, the catalytic activation scenario we presented transforms a biseparable state (hence completely resourceless) into a GME state. This is similar to the recent demonstration of catalytic activation of quantum nonlocality \cite{bavaresco}, whereby an entangled state admitting a local model is catalytically transformed into a state violating a Bell inequality. Finally, our work also complements recent ones demonstrating the multi-copy activation of GMNL~\cite{Contreras22, miethlinger26}, and catalytic activation of GMNL at the level of nonlocal boxes~\cite{ulu2026}.

Finally, our results question the operational meaning of the (standard) definition of GME. Indeed, one may argue that stability under CLO is a natural requirement. In a similar spirit, the authors of Ref.~\cite{Palazuelos22} argued that the lack of stability of GME in the multi-copy scenario is problematic. Alternative definitions could be considered, such as the notion of genuine multipartite network entanglement proposed in Ref.~\cite{Navascues20} (see also \cite{Kraft_2021})
which is multi-copy stable by construction. An interesting open question is whether this notion is also stable under CLO.

\acknowledgements
We thank Jessica Bavaresco, Patryk Lipka-Bartosik, and Bora Ulu for useful discussions.~We acknowledge funding by the Swiss National Science Foundations (projects 219366 and 236580). This work was funded with a grant from the Foundation for the University of Geneva.

\bibliographystyle{apsrev4-2}
\bibliography{ref}

\onecolumngrid
\newpage
\section*{Appendix}

\appendix
\section{\label{appendix1}General proof of \Cref{lem1} and continuous set of catalysts}

\THMOne*

\begin{proof}
The proof is independent of the number of parties $N$. The initial system state is $
    \rho_S=\sum_i \alpha_i \rho_i\otimes[i]$,
and the state $\rho_{S}\otimes \omega_C$ is therefore a mixture of $n^2$ elements, as represented in Table~\ref{tableapp1}.
 
We now proceed to construct the set of local operations $\Lambda$ leading to the desired transformation. The first step is to enlarge the system by local preparing of $n-1$ product states $\sigma^{\otimes (n-1)}$. The remaining operations are only made of local swaps and flags relabeling, as described below in steps (i)-(iv).

\begin{itemize}
    \item[\textbf{(i)}] When parties read flag $[1]_S\otimes [1]_C$, they build the tensor product state $\bigotimes_{k=1}^n \rho_k$ by swapping $\sigma^{\otimes (n-1)}$ with the state $\bigotimes_{k=2}^n \rho_k$ of the catalyst, 
 \begin{equation}
     \rho_1\otimes\sigma^{\otimes (n-1)}\otimes\left(\bigotimes_{k=2}^n \rho_k\right)\rightarrow \left(\bigotimes_{k=1}^n \rho_k\right)\otimes\sigma^{\otimes (n-1)}.
 \end{equation}
After that step, the tensor product branch of the final system state is reached. Note that this happen with probability $1/C_\alpha$. Additional operations are necessary to recover the catalyst in the second marginal. To do so, one needs to recover $n$ times each elements of the catalyst. Note that the  $n$ states $\sigma^{\otimes (n-1)}$ are already recovered in entries $[1]\otimes[1]$ and $[x] \otimes [n]$ for $x\neq n$.

\item[\textbf{(ii)}] Entries $[1]_S\otimes [x]_C$ for $x\neq1$ already give $n-1$ states $\bigotimes_{k=2}^n \rho_k$ in the catalyst subspace. Therefore no actions are required when parties measure these flags. The $n$-th state $\bigotimes_{k=2}^n \rho_k$ is recovered in the diagonal entry $[2]_S\otimes[2]_C$ by swapping $\rho_2$ system's state and a $\sigma$ of the catalyst:
\begin{equation}
     \rho_2\otimes\sigma^{\otimes (n-1)}\otimes\sigma \otimes \left(\bigotimes_{k=3}^n \rho_k\right)\rightarrow \sigma^{\otimes n}\otimes\left(\bigotimes_{k=2}^n \rho_k\right).
 \end{equation}
 
\item[\textbf{(iii)}] More generally, for all $j\in\{2,\dots,n-1\}$, entries $[j]_S\otimes [x]_C$ for $x\neq j$ give $n-1$ states $\sigma^{\otimes (j-1)} \otimes \left(\bigotimes_{k=j+1}^n \rho_k\right)$ in the catalyst subspace, and no actions are required when parties measure these flags. The $n$-th state $\sigma^{\otimes (j-1)} \otimes \left(\bigotimes_{k=j+1}^n \rho_k\right)$ is recovered in the diagonal entry $[j+1]_S\otimes[j+1]_C$ by swapping $\rho_{j+1}$ system's state and one $\sigma$ of the catalyst:
\begin{equation}
     \rho_{j+1} \otimes\sigma^{\otimes (n-1)}\otimes\sigma^{\otimes j} \otimes \left(\bigotimes_{k=j+2}^n \rho_k\right)\rightarrow \sigma^{\otimes n}\otimes\sigma^{(j-1)}\otimes \left(\bigotimes_{k=j+1}^n \rho_k\right).
 \end{equation}
\item[\textbf{(iv)}] At each step, parties add a copy of the catalyst flag in the system $[i]_S\otimes[j]_C\rightarrow [i]_{S'}\otimes[j]_{S'}\otimes [j]_C$ and perform the cyclic relabel of the diagonal entries $[i]_{S'}\otimes [i]_{S'}\otimes[i]_C\rightarrow [i-1]_{S'}\otimes [i-1]_{S'}\otimes[i-1]_C$ with $[0]\equiv[n]$.
\end{itemize}

After these operations, the final state $\tau_{S'C}$ is such that the partial trace $\tau_{S'} :=\text{Tr}_C(\tau_{S'C})$ is (up to a cyclic flag relabeling) of the form given in Eq.~\eqref{target}. For all $j\in\{1,\dots,n\}$, each term $\sigma^{\otimes (j-1)} \otimes \left(\bigotimes_{k=j+1}^n \rho_k\right)$ in the marginal state $\text{Tr}_{S'}(\tau_{S'C})$ appears with weight $\sum_{k\neq j} \frac{\alpha_k}{C_\alpha \alpha_j} + \frac{1}{C_\alpha} = \frac{1}{C_\alpha}\frac{1}{\alpha_j}$. It thus has the exact same weights as the catalyst given in Eq.~\eqref{catalyst}, and the catalytic marginal conditions of Eq.~\eqref{cdt} are thus satisfied.
\end{proof}

\begin{table*}[t]
\setlength{\tabcolsep}{6pt}
\renewcommand{\arraystretch}{1.7} 
\small 
\begin{tabular*}{\columnwidth}{@{\extracolsep{\fill}}c|cccccc}
\multirowcell{1.2}[3ex]{\diagbox[height=3\line]{\rlap{\raisebox{1ex}{(Cat)  }}}{\raisebox{-2ex}{Syst}}} & $\rho_1$ &
$\rho_2$ & $\cdots$ & $\cdots$ & $\cdots$ & 
$\rho_n$\\
\hline
$(\bigotimes_{k=2}^n\rho_k)$ & $\rho_1\otimes (\bigotimes_{k=2}^n\rho_k)$ &
$\rho_2\otimes (\bigotimes_{k=2}^n\rho_k)$ & $\cdots$ & $\cdots$ & $\cdots$ & 
$\rho_n\otimes (\bigotimes_{k=2}^n\rho_k)$\\[4pt]
$(\sigma \otimes \bigotimes_{k=3}^n\rho_k)$& $\rho_1\otimes (\sigma \otimes \bigotimes_{k=3}^n\rho_k)$ &
$\rho_2\otimes (\sigma \otimes \bigotimes_{k=3}^n\rho_k)$ & $\cdots$ & $\cdots$ & $\cdots$ & 
$\rho_n\otimes (\sigma \otimes \bigotimes_{k=3}^n\rho_k)$\\[4pt]
$\vdots$ & $\vdots$ &
$\vdots$ & $\ddots$ & $\ddots$ & $\ddots$ & 
$\vdots$\\[4pt]
$\vdots$ & $\vdots$ &
$\vdots$ & $\ddots$ & $\rho_j \otimes (\sigma^{\otimes (j-1)}\otimes \left(
\bigotimes_{k=j+1}^n \rho_k \right))$ & $\ddots$ & 
$\vdots$\\[4pt]
$\vdots$ & $\vdots$ &
$\vdots$ & $\ddots$ & $\ddots$ & $\ddots$ & 
$\vdots$\\[4pt]
$(\sigma^{(n-2)} \otimes \rho_n)$& $\rho_1\otimes (\sigma^{(n-2)} \otimes \rho_n)$ &
$\rho_2\otimes (\sigma^{(n-2)} \otimes \rho_n)$ & $\cdots$ & $\cdots$ & $\cdots$ & 
$\rho_n\otimes (\sigma^{(n-2)} \otimes \rho_n)$\\[4pt]
$(\sigma^{(n-1)})$& $\rho_1\otimes (\sigma^{(n-1)})$ &
$\rho_2\otimes (\sigma^{(n-1)})$ & $\cdots$ & $\cdots$ & $\cdots$ & 
$\rho_n\otimes (\sigma^{(n-1)})$
\end{tabular*}
\caption{Table representation of the state $\rho_S\otimes \omega_C$. %
Each entry is the tensor product of one system state with one catalyst state, carrying weight $\frac{\alpha_i}{C_\alpha \alpha_j}$ and a unique flag $[i]_S\otimes[j]_C$ that allows it to be distinguished from the others by local flag reading.}
\label{tableapp1}
\end{table*}

\begin{corollary}[Permuted catalysts]
    \label{lemApendix1}
Let $\mathcal{S}_n$ be the set of permutations of $\{1,\dots,n\}$. For any permutation $\Gamma \in \mathcal{S}_n$, the catalytic transformation of~\Cref{lem1} can be performed using the following permuted catalyst:
\begin{equation}
    \omega_C(\Gamma) = \frac{1}{C_\alpha}\sum_{j=1}^{n} \frac{1}{\alpha_j} 
\sigma^{\otimes (j-1)}\otimes\left(
\bigotimes_{k=j+1}^n \rho_{\Gamma(k)}\right) \otimes [j],
\label{permcatalyst}
\end{equation}
\end{corollary}

\begin{proof}
    The proof comes from the fact that up to reordering, both the initial state of Eq.~\eqref{mixturegeneral} and the target state in Eq.~\eqref{target} are invariant under permutation of the indices. 
\end{proof}

\begin{corollary}[Continuous family of catalysts]
    \label{lemApendix2}
The catalytic transformation of~\Cref{lem1} can be performed using any flagged convex mixtures of the $n!$ permuted catalysts in Eq.~\eqref{permcatalyst}. Equivalently, denoting $\Gamma_1,\dots,\Gamma_{n!}$ the permutations of $\mathcal{S}_n$, for any family of coefficients $\beta_1,\dots,\beta_{n!} \geq 0$, such that $\sum_r \beta_r = 1$, the catalytic transformation of~\Cref{lem1} can be performed using the following catalyst:
\begin{equation}
    \omega_C(\bm \beta) = \sum_{r=1}^{n!}  \beta_r
\, \omega_C(\Gamma_r) \otimes [r].
\label{contcatalyst}
\end{equation}
\end{corollary}

\begin{proof}
    The local operations for the catalytic transformation is to first read the local flags $[r]$ to know which catalyst is shared, and then use the appropriate set of operations, following \Cref{lemApendix1}.
\end{proof}

\section{\label{appendix2}General proof of Theorem~\ref{lem2}}

\THMtwo*

\begin{proof}
    Be $(\omega_{01},\Lambda_{01})$ that operates without lost of generalities as
    \begin{equation}
        \Lambda_{01}(\rho_0\otimes \omega_{01})=p_1\rho_1\otimes[1]_S\otimes \omega_{01}^{(1)}\otimes [1]_C+\sum_{i> 1}p_i\sigma_0^{(i)}\otimes[i]_S\otimes \omega_{01}^{(i)}\otimes [i]_C.
    \end{equation}
Where $\omega_{01}=\sum_i\xi_i\omega_{01}^{(i)}\otimes [i]$. This state verifies the catalytic marginal conditions of Eq.~\eqref{cdt}. One thus applies that first CLO on the full state $\rho_0\otimes \omega_{01}\otimes \omega_{12}$
\begin{equation}
    \Lambda_{01}(\rho_0\otimes \omega_{01}\otimes\omega_{12})=\Lambda_{01}(\rho_0\otimes \omega_{01})\otimes\omega_{12}=p_1\rho_1\otimes[1]_S\otimes \omega_{01}^{(1)}\otimes [1]_C\otimes\omega_{12}+\sum_{i> 1}p_i\sigma_0^{(i)}\otimes[i]_S\otimes \omega_{01}^{(i)}\otimes [i]_C\otimes\omega_{12}.
\end{equation}
We now operate with the two following local operations:
\begin{itemize}
\item[\textbf{i)}] When the flag $[1]_S\otimes [1]_C$ is measured, We swap the catalyst subspaces and we apply the set $\Lambda_{12}$ on $\rho_1\otimes \omega_{12}$. That branch thus becomes
\begin{equation}
    p_1\left(  \left(q_1\rho_2\otimes\tilde{[2]}_S\otimes \omega_{12}^{(2)}\otimes \tilde{[2]}_C+\sum_{i> 1}q_i\sigma_1^{(i)}\otimes\tilde{[i]}_S\otimes \omega_{12}^{(i)}\otimes \tilde{[i]}_C\right)   \otimes[1]_S\otimes [1]_C\otimes\omega_{01}^{(1)}\right)
\end{equation}
The partial trace over the system is 
\begin{equation}
    p_1q_1\otimes \omega_{12}^{(2)}\otimes \tilde{[2]}_C\otimes [1]_C\otimes\omega_{01}^{(1)}+p_1\sum_{i> 1}q_i \omega_{12}^{(i)}\otimes \tilde{[i]}_C \otimes [1]_C\otimes\omega_{01}^{(1)}
\end{equation}
While the other partial trace is
\begin{equation}
    p_1q_2\rho_2\otimes\tilde{[2]}_S\otimes[1]_S+p_1\sum_{i> 1}q_i\sigma_1^{(i)}\otimes\tilde{[i]}_S   \otimes[1]_S
\end{equation}
\item[\textbf{ii)}] When any other flag is measured return the state unchanged with an additional orthogonal flag $\tilde{[j]}_S$. The partial trace over the system is thus
\begin{equation}
    \sum_{i> 1}\omega_{01}^{(i)}\otimes [i]_C\otimes\omega_{12}
\end{equation}
While the other one is
\begin{equation}
    \sum_{i> 1}p_i\sigma_0^{(i)}\otimes[i]_S\otimes \tilde{[j]}_S
\end{equation}
The sum of the partial traces of these two steps gives the full partial traces. The catalytic marginal is thus exactly $\omega_{01}\otimes\omega_{12}$. On the other hand, the second partial trace is, after flag relabel
\begin{equation}
    \tau_S=p_1q_2\rho_2\otimes[0]+\sum_i\xi_{i\neq 0}\sigma_2^{(i)}\otimes[i].
\end{equation}
\end{itemize}
The catalytic marginal conditions of Eq.~\eqref{cdt} are thus satisfied. 
\end{proof}

\section{Proof of \Cref{thm2}.}\label{lemtechniqueAp}

\THMThree*

\begin{proof}
    We proceed by contradiction and consider that the distribution $p$ is not GMNL, i.e.~can be written as 
    \begin{equation}
    \label{distribcntradiction}p(\vect{a}|\vect{x})=\sum_\mathcal{G}\eta_\mathcal{G}q_\mathcal{G}(\vect{a}_\mathcal{G},\vect{a}_{\bar{\mathcal{G}}}|\vect{x}_\mathcal{G},\vect{x}_{\bar{\mathcal{G}}}),
    \end{equation}
    where $q_\mathcal{G}(\vect{a}_\mathcal{G},\vect{a}_{\bar{\mathcal{G}}}|\vect{x}_\mathcal{G},\vect{x}_{\bar{\mathcal{G}}})$ are local distributions over the bi-partition $\mathcal{G}|\bar{\mathcal{G}}$ and with $\eta_\mathcal{G}\geq 0$, $\sum_\mathcal{G}\eta_\mathcal{G}=1$.
    
    The hypothesis directly implies that for any bi-partition $\mathcal{G}|\bar{\mathcal{G}}$,
    \begin{equation}
        \label{condition} p(\vect{a}_\mathcal{G}\in \bm{M},\vect{a}_{\bar{\mathcal{G}}}\notin\bm{M}|\vect{x})=p(\vect{a}_\mathcal{G}\notin \bm{M},\vect{a}_{\bar{\mathcal{G}}}\in\bm{M}|\vect{x})=0.
    \end{equation}
    By convexity, all the distributions $q_\mathcal{G}$ in Eq.~\eqref{distribcntradiction} also satisfies Eq.~\eqref{condition}. Together with the no-signaling condition (NS) that all local distributions $q_\mathcal{G}$ fulfill, this implies that for all $\mathcal{G}$ and all $\vect{x}_{\mathcal{G}}, \vect{x}_{\bar{\mathcal{G}}}$:
    \begin{equation}
    \begin{split}
        q_\mathcal{G}(\bm{M}|\vect{x}_\mathcal{G},\vect{x}_{\bar{\mathcal{G}}})& = \sum_{a_{\mathcal{G}}, a_{\bar{\mathcal{G}}}\in \bm{M}} q_\mathcal{G}(\vect{a}_{\mathcal{G}}\in\bm{M},\vect{a}_{\bar{\mathcal{G}}}\in\bm{M}|\vect{x}_\mathcal{G},\vect{x}_{\bar{\mathcal{G}}}) \\
        & \overset{\eqref{condition}}{=} \sum_{a_{\mathcal{G}} \in \bm{M}} \sum_{a_{\bar{\mathcal{G}}}}q_\mathcal{G}(\vect{a}_\mathcal{G}\in\bm{M},\vect{a}_{\bar{\mathcal{G}}}|\vect{x}_\mathcal{G},\vect{x}_{\bar{\mathcal{G}}})\\
        & \overset{\text{(NS)}}{=} \sum_{a_{\mathcal{G}} \in \bm{M}} q_\mathcal{G}(\vect{a}_\mathcal{G}\in\bm{M}|\vect{x}_\mathcal{G}).
    \end{split}
    \end{equation}
Especially it implies that $q_\mathcal{G}(\bm{M}|\vect{x}_\mathcal{G},\vect{x}_{\bar{\mathcal{G}}})$ doesn't depend on $\vect{x}_{\bar{\mathcal{G}}}$. Likewise, on can prove it doesn't depend on $\vect{x}_{\mathcal{G}}$, and thus with define $q_\mathcal{G}(\bm{M}|\vect{x}_\mathcal{G},\vect{x}_{\bar{\mathcal{G}}})=: k_\mathcal{G}$ which is indeed independent of $\vect{x}_\mathcal{G},\vect{x}_{\bar{\mathcal{G}}}$.

On the other hand, using the law of complete probabilities, one can write the distribution $q_\mathcal{G}$ as:
\begin{equation} \label{eq:qGdecomp}
    \begin{split}
        q_\mathcal{G}(\vect{a}_\mathcal{G},\vect{a}_{\bar{\mathcal{G}}}|\vect{x}_\mathcal{G},\vect{x}_{\bar{\mathcal{G}}})&=q_\mathcal{G}((\vect{a}_\mathcal{G}, \vect{a}_{\bar{\mathcal{G}}}) \cap (\vect{a}_{\mathcal{G}}\in\bm{M},\vect{a}_{\bar{\mathcal{G}}}\in\bm{M})|\vect{x}_\mathcal{G},\vect{x}_{\bar{\mathcal{G}}}) \\
        & \qquad + q_\mathcal{G}((\vect{a}_\mathcal{G}, \vect{a}_{\bar{\mathcal{G}}})\cap (\vect{a}_{\mathcal{G}}\notin\bm{M},\vect{a}_{\bar{\mathcal{G}}}\in\bm{M})|\vect{x}_\mathcal{G},\vect{x}_{\bar{\mathcal{G}}}) \\
        & \qquad + q_\mathcal{G}((\vect{a}_\mathcal{G}, \vect{a}_{\bar{\mathcal{G}}})\cap (\vect{a}_{\mathcal{G}}\in\bm{M},\vect{a}_{\bar{\mathcal{G}}}\notin\bm{M})|\vect{x}_\mathcal{G},\vect{x}_{\bar{\mathcal{G}}}) \\
        & \qquad + q_\mathcal{G}((\vect{a}_\mathcal{G}, \vect{a}_{\bar{\mathcal{G}}})\cap (\vect{a}_{\mathcal{G}}\notin\bm{M},\vect{a}_{\bar{\mathcal{G}}}\notin\bm{M})|\vect{x}_\mathcal{G},\vect{x}_{\bar{\mathcal{G}}}) \\
     & \overset{\eqref{condition}}{=} q_\mathcal{G}((\vect{a}_\mathcal{G}, \vect{a}_{\bar{\mathcal{G}}}) \cap(\vect{a}_{\mathcal{G}}\in\bm{M},\vect{a}_{\bar{\mathcal{G}}}\in\bm{M})|\vect{x}_\mathcal{G},\vect{x}_{\bar{\mathcal{G}}}) \\
        & \qquad + q_\mathcal{G}((\vect{a}_\mathcal{G}, \vect{a}_{\bar{\mathcal{G}}})\cap(\vect{a}_{\mathcal{G}}\notin\bm{M},\vect{a}_{\bar{\mathcal{G}}}\notin\bm{M})|\vect{x}_\mathcal{G},\vect{x}_{\bar{\mathcal{G}}}) \\
&=k_\mathcal{G} \cdot r_\mathcal{G}(\vect{a}_\mathcal{G},\vect{a}_{\bar{\mathcal{G}}}|\mathbf{x}_\mathcal{G},\vect{x}_{\bar{\mathcal{G}}})
    +q_\mathcal{G}((\vect{a}_\mathcal{G}, \vect{a}_{\bar{\mathcal{G}}})\cap(\vect{a}_{\mathcal{G}}\notin\bm{M},\vect{a}_{\bar{\mathcal{G}}}\notin\bm{M})|\vect{x}_\mathcal{G},\vect{x}_{\bar{\mathcal{G}}}),
    \end{split}
\end{equation}
where we defined a new local probability distribution
\begin{equation}\label{defr}
\begin{split}
    & r_\mathcal{G}(\vect{a}_\mathcal{G},\vect{a}_{\bar{\mathcal{G}}}|\mathbf{x}_\mathcal{G},\vect{x}_{\bar{\mathcal{G}}}) := q_\mathcal{G}(\vect{a}_\mathcal{G}, \vect{a}_{\bar{\mathcal{G}}}|\vect{x}_\mathcal{G},\vect{x}_{\bar{\mathcal{G}}}, \vect{a}_{\mathcal{G}}\in\bm{M},\vect{a}_{\bar{\mathcal{G}}}\in\bm{M}).
\end{split}
\end{equation}
Notice that since it can by definition only take outputs in $\bm{M}$, it satisfies 
\begin{equation} \label{eq:rGprop}
    \sum_{\vect{a}_{\mathcal{G}},\vect{a}_{\bar{\mathcal{G}}}\in\bm{M}} r_\mathcal{G}(\vect{a}_\mathcal{G},\vect{a}_{\bar{\mathcal{G}}}|\mathbf{x}_\mathcal{G},\vect{x}_{\bar{\mathcal{G}}}) = \sum_{\vect{a}_{\mathcal{G}},\vect{a}_{\bar{\mathcal{G}}}} r_\mathcal{G}(\vect{a}_\mathcal{G},\vect{a}_{\bar{\mathcal{G}}}|\mathbf{x}_\mathcal{G},\vect{x}_{\bar{\mathcal{G}}})=1.
\end{equation}
    
Inserting Eq.~\eqref{eq:qGdecomp} into Eq.~\eqref{distribcntradiction}, leads to 
    \begin{equation}
        p(\vect{a}|\vect{x})=\sum_\mathcal{G}\eta_\mathcal{G}\big[ k_\mathcal{G} \cdot r_\mathcal{G}(\vect{a}_\mathcal{G},\vect{a}_{\bar{\mathcal{G}}}|\mathbf{x}_\mathcal{G},\vect{x}_{\bar{\mathcal{G}}})
    +q_\mathcal{G}(\vect{a}_\mathcal{G}, \vect{a}_{\bar{\mathcal{G}}}, \vect{a}_{\mathcal{G}}\notin\bm{M},\vect{a}_{\bar{\mathcal{G}}}\notin\bm{M}|\vect{x}_\mathcal{G},\vect{x}_{\bar{\mathcal{G}}}) \big].
    \end{equation}
    Equalizing with Eq.\eqref{taufinal}, for all $\vect{a}\in \bm{M}$, it holds that  
    \begin{equation}
        p_\text{GMNL}(\vect{a}\in \bm{M}|\vect{x})=\sum_\mathcal{G}\tilde{\eta}_\mathcal{G} r_\mathcal{G}(\vect{a}_\mathcal{G},\vect{a}_{\bar{\mathcal{G}}}|\mathbf{x}_\mathcal{G},\vect{x}_{\bar{\mathcal{G}}}),
    \end{equation}
    where $\tilde{\eta}_\mathcal{G}:=\eta_\mathcal{G}k_\mathcal{G}/\kappa\geq 0$. Summing over all $a\in \bm{M}$, and using Eq.~\eqref{eq:rGprop}, we obtain $\sum_\mathcal{G} \tilde{\eta}_\mathcal{G} = 1$, i.e.~that $\tilde{\eta}$ is a normalized probability distribution over $\{\mathcal{G}\}$, which contradicts the assumption on $p_\text{GMNL}$. 
\end{proof}
\newpage
\section{Visualization of the catalyst state for the transformation in Fig.~\ref{image}.} \label{app:catalystfig1}

\begin{figure}[h]
    \centering
    \includegraphics[trim=6cm 10.5cm 6.5cm 6cm, clip=true, width=1\linewidth]{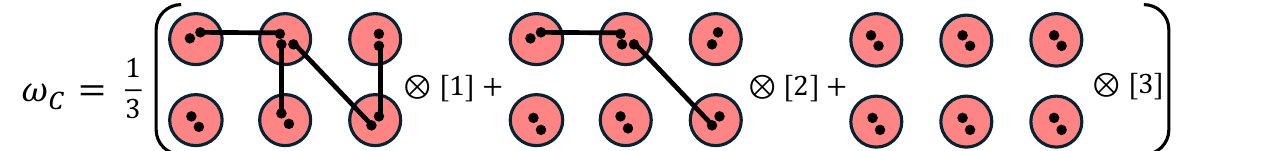}
    \caption{Graphic representation of one of the catalysts that enables the transformation represented in Fig.~\ref{image}. The latter is of the form $\omega_C=\frac{1}{3}(\rho_2\otimes\rho_3\otimes[1]+\sigma\otimes\rho_3\otimes[2]+\sigma\otimes\sigma\otimes[3])$.}
\end{figure}

\section{Visualization of different possible catalysts \label{exampleapendix}}
\begin{figure}[h]
    \centering
    \includegraphics[trim=0cm 4cm 9cm 1cm, clip=true, width=1\linewidth]{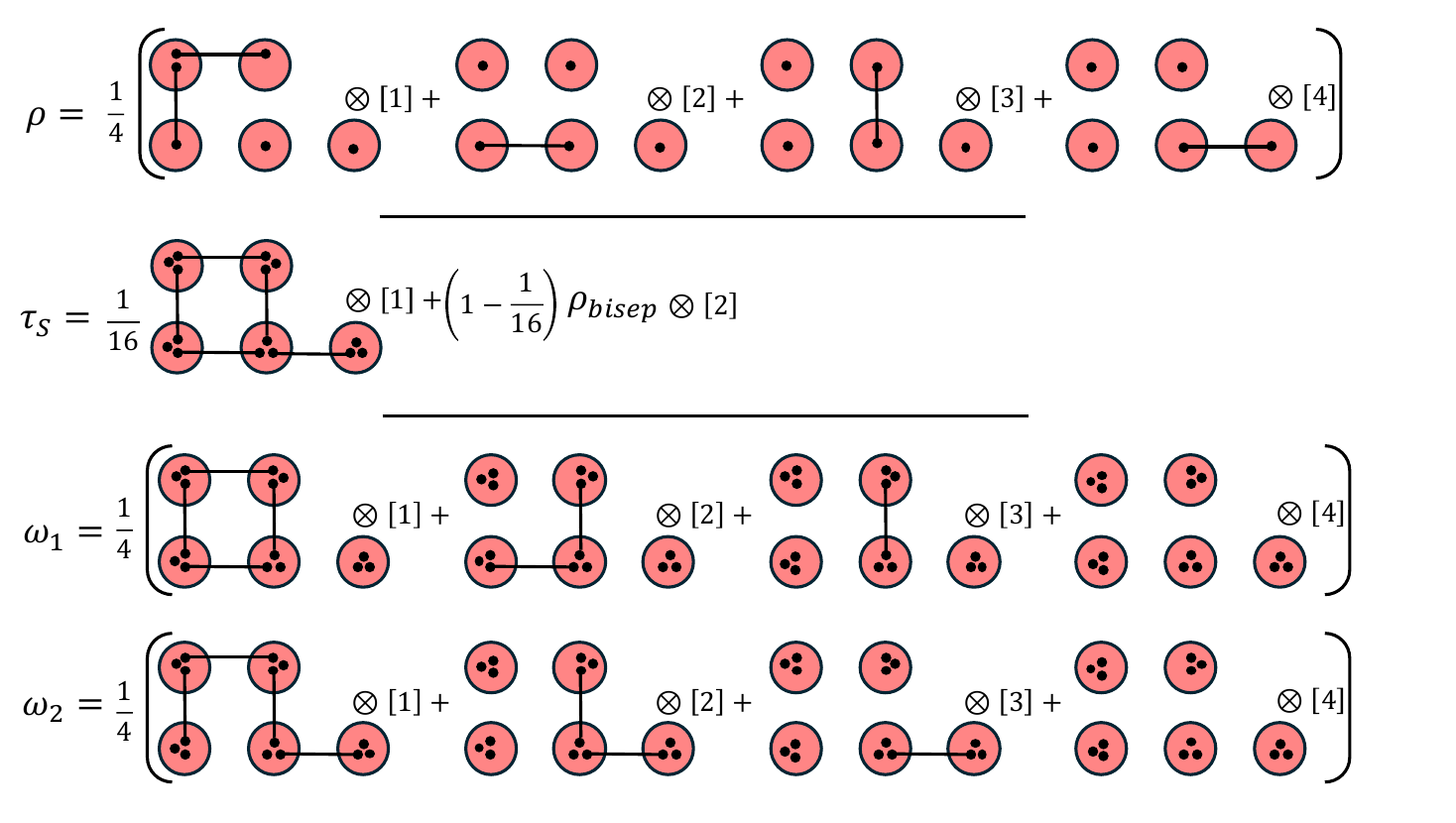}
    \caption{Graphic representation of a biseparable state $\rho$ that is catalytically transformable via \Cref{lem1} into the target state~$\tau_S$ that is $5$-partite GMNL. Following \Cref{lemApendix1}, $\omega_1$ and $\omega_2$ are two of the $16$ possible canonical catalysts for such transformation. However, $\omega_1$ is partition separable over the 5 parties, and thus not GME nor GMNL, while $\omega_2$ is already $5$-partite GMNL.}
    \label{appendixfig}
\end{figure}

\end{document}